\documentclass[a4paper,UKenglish,cleveref, autoref]{lipics-v2019}
\usepackage{amsmath,hyperref,lineno,amssymb,amsthm,mathtools,url,verbatim}
\usepackage{pgfplots}
\pgfplotsset{grid style={red}}
\usepackage{filecontents}
\usepackage{microtype}
\usepackage{float}
\usepackage{tikz}
\usetikzlibrary{matrix}
\usepackage{tikz-cd}
\usepackage{algorithmic, algorithm2e}

\title{The mergegram of a dendrogram and its stability}
\titlerunning{The mergegram of a dendrogram and its stability}

\author{Yury Elkin}{Materials Innovation Factory and Computer Science department, University of Liverpool, UK}{yura.elkin@gmail.com}{}{}

\author{Vitaliy Kurlin}{Materials Innovation Factory and Computer Science department, University of Liverpool, UK}{vitaliy.kurlin@gmail.com}{http://kurlin.org}{}

\authorrunning{Y.Elkin et. al.}

\Copyright{Yury Elkin, Vitaliy Kurlin}



\category{}

\relatedversion{}
\nolinenumbers
\supplement{}

\funding{The authors were supported by the £3.5M EPSRC grant EP/R018472/1 (2018-2023) 
}

\EventEditors{Javier Esparza and Daniel Kr{\'a}l'}
\EventNoEds{2}
\EventLongTitle{45th International Symposium on Mathematical Foundations of Computer Science (MFCS 2020)}
\EventShortTitle{MFCS 2020}
\EventAcronym{MFCS}
\EventYear{2020}
\EventDate{August 24--28, 2020}
\EventLocation{Prague, Czech Republic}
\EventLogo{}
\SeriesVolume{170}
\ArticleNo{56}
\nolinenumbers 

\theoremstyle{definition}
\newtheorem{dfn}{Definition}[section]
\newtheorem{thm}[dfn]{Theorem}
\newtheorem{exa}[dfn]{Example}
\newtheorem{myclaim}[dfn]{Claim}
\newtheorem{alg}[dfn]{Algorithm}

\newtheorem{lem}[dfn]{Lemma}

\newcommand{\R}{\mathbb{R}}
\newcommand{\Z}{\mathbb{Z}}
\newcommand{\PS}{\mathbb{P}}
\newcommand{\U}{\mathbb{U}}
\newcommand{\V}{\mathbb{V}}
\newcommand{\W}{\mathbb{W}}
\newcommand{\I}{\mathbb{I}}
\newcommand{\id}{\mathrm{id}}
\newcommand{\life}{\mathrm{life}}
\newcommand{\MG}{\mathrm{MG}}
\newcommand{\PD}{\mathrm{PD}}
\newcommand{\HD}{\mathrm{HD}}
\newcommand{\BD}{\mathrm{BD}}
\newcommand{\ID}{\mathrm{ID}}
\newcommand{\MST}{\mathrm{MST}}
\newcommand{\de}{\delta}
\newcommand{\De}{\Delta}

\newcommand{\bs}{\hfill $\blacksquare$}

\begin{document}
\maketitle

\begin{abstract}
This paper extends the key concept of persistence within Topological Data Analysis (TDA) in a new direction.
TDA quantifies topological shapes hidden in unorganized data such as clouds of unordered points.
In the 0-dimensional case the distance-based persistence is determined by a single-linkage (SL) clustering of a finite set in a metric space.
Equivalently, the 0D persistence captures only edge-lengths of a Minimum Spanning Tree (MST).
Both SL dendrogram and MST are unstable under perturbations of points.
We define the new stable-under-noise mergegram, which outperforms
previous isometry invariants on a classification of point clouds by PersLay.
\end{abstract}

\section{Introduction: motivations and overview of the new results}
\label{sec:intro}

TDA is now expanding towards machine learning and statistics due to stability that was proved in a very general form by Chazal et al. \cite{chazal2016structure}.
The key idea of TDA is to view a given cloud of points across all scales $s$, e.g. by blurring given points to balls of a variable radius $s$.
The resulting evolution of topological shapes is summarized by a persistence diagram.
\medskip

\begin{exa}
\label{exa:5-point_line}
Fig.~\ref{fig:5-point_line} illustrates the key concepts (before formal definitions) for the point set $A = \{0,4,6,9,10\}$ in the real line $\R$.
Imagine that we gradually blur original data points by growing balls of the same radius $s$ around the given points.
The balls of the closest points $9,10$ start overlapping at the scale $s=0.5$ when these points merge into one cluster $\{9,10\}$.
This merger is shown by blue arcs joining at the node at $s=0.5$ in the single-linkage dendrogram, see the bottom left picture in Fig.~\ref{fig:5-point_line} and more details in Definition~\ref{dfn:sl_clustering}.

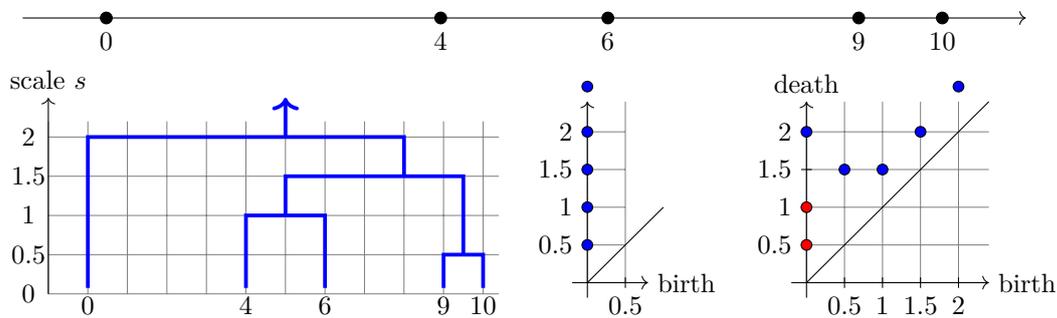
\begin{figure}[H]
\centering
\begin{tikzpicture}[scale = 1.1]
  \draw[->] (-1,0) -- (11,0) node[right]{} ;
   \foreach \x/\xtext in {0, 4, 6, 9, 10}
    \draw[shift={(\x,0)}] (0pt,2pt) -- (0pt,-2pt) node[below] {$\xtext$};   
   \filldraw (0,0) circle (2pt);
   \filldraw (4,0) circle (2pt);
   \filldraw (6,0) circle (2pt);
   \filldraw (9,0) circle (2pt);
   \filldraw (10,0) circle (2pt);
\end{tikzpicture}

\begin{tikzpicture}[scale = 0.52][sloped]
\draw[style=help lines,step = 1] (-1,0) grid (10.4,4.4);
\draw [->] (-1,0) -- (-1,5) node[above] {scale $s$};
\foreach \i in {0,0.5,...,2}{ \node at (-1.5,2*\i) {\i}; }
\node (a) at (0,-0.3) {0};
\node (b) at (4,-0.3) {4};
\node (c) at (6,-0.3) {6};
\node (d) at (9,-0.3) {9};
\node (e) at (10,-0.3) {10};
\node (x) at (5,5) {};
\node (de) at (9.5,1){};
\node (bc) at (5.0,2){};
\node (bcde) at (8.0,3){};
\node (all) at (5.0,4){};
\draw [line width=0.5mm, blue ] (a) |- (all.center);
\draw [line width=0.5mm, blue ] (b) |- (bc.center);
\draw [line width=0.5mm, blue ] (c) |- (bc.center);
\draw [line width=0.5mm, blue ] (d) |- (de.center);
\draw [line width=0.5mm, blue ] (e) |- (de.center);
\draw [line width=0.5mm, blue ] (de.center) |- (bcde.center);
\draw [line width=0.5mm, blue ] (bc.center) |- (bcde.center);
\draw [line width=0.5mm, blue ] (bcde.center) |- (all.center);
\draw [line width=0.5mm, blue ] [->] (all.center) -> (x.center);
\end{tikzpicture}
\hspace*{1mm}
\begin{tikzpicture}[scale = 1.0]
  \draw[style=help lines,step = 0.5] (0,0) grid (0.5,2.4);
  \draw[->] (-0.2,0) -- (0.8,0) node[right] {birth};
  \draw[->] (0,-0.2) -- (0,2.4) node[above] {};	
  \draw[-] (0,0) -- (1,1) node[right]{};
  \foreach \x/\xtext in {0.5/0.5}
    \draw[shift={(\x,0)}] (0pt,2pt) -- (0pt,-2pt) node[below] {$\xtext$};
  \foreach \y/\ytext in {0.5/0.5,  1/1, 1.5/1.5, 2.0/2}
    \draw[shift={(0,\y)}] (2pt,0pt) -- (-2pt,0pt) node[left] {$\ytext$};  
   \filldraw [fill=blue] (0,0.5) circle (2pt);
   \filldraw [fill=blue] (0.0,1) circle (2pt);
   \filldraw [fill=blue] (0,1.5) circle (2pt);
   \filldraw [fill=blue] (0,2) circle (2pt);
   \filldraw [fill=blue] (0,2.6) circle (2pt);
\end{tikzpicture}
\hspace*{1mm}
\begin{tikzpicture}[scale = 1.0]
  \draw[style=help lines,step = 0.5] (0,0) grid (2.4,2.4);
  \draw[->] (-0.2,0) -- (2.4,0) node[right] {birth};
  \draw[->] (0,-0.2) -- (0,2.4) node[above] {death};	
  \draw[-] (0,0) -- (2.4,2.4) node[right]{};
  \foreach \x/\xtext in {0.5/0.5, 1/1, 1.5/1.5, 2.0/2}
    \draw[shift={(\x,0)}] (0pt,2pt) -- (0pt,-2pt) node[below] {$\xtext$};
  \foreach \y/\ytext in {0.5/0.5,  1/1, 1.5/1.5, 2.0/2}
    \draw[shift={(0,\y)}] (2pt,0pt) -- (-2pt,0pt) node[left] {$\ytext$};  
   \filldraw [fill=red] (0,0.5) circle (2pt);
   \filldraw [fill = red] (0.0,1) circle (2pt);
   \filldraw [fill=blue] (0.0,2) circle (2pt);
   \filldraw [fill=blue] (0.5,1.5) circle (2pt);
   \filldraw [fill=blue] (1.0,1.5) circle (2pt);
   \filldraw [fill=blue] (1.5, 2.0) circle (2pt);
   \filldraw [fill=blue] (2, 2.6) circle (2pt);
\end{tikzpicture}
\caption{\textbf{Top}: the 5-point cloud $A = \{0,4,6,9,10\}\subset\R$.
\textbf{Bottom} from left to right: single-linkage dendrogram $\De_{SL}(A)$ from Definition~\ref{dfn:sl_clustering}, the 0D persistence diagram $\PD$ from Definition~\ref{dfn:persistence_diagram} and the new mergegram $\MG$ from Definition~\ref{dfn:mergegram}, where the red color shows dots of multiplicity 2.}
\label{fig:5-point_line}
\end{figure}

The persistence diagram $\PD$ in the bottom middle picture of Fig.~\ref{fig:5-point_line} represents this merger by the dot $(0,0.5)$ meaning that a singleton cluster of (say) point $9$ was born at the scale $s=0$ and then died later at $s=0.5$ (by merging into another cluster of point 10), see details in Definition~\ref{dfn:sl_clustering}.
When two clusters $\{4,6\}$ and $\{9,10\}$ merge at $s=1.5$, this event was previously encoded in the persistence diagram by the single dot $(0,1.5)$ meaning that one cluster inherited from (say) point 10 was born at $s=0$ and has died at $s=1.5$.
\medskip

For the same merger, the new mergegram in the bottom right picture of Fig.~\ref{fig:5-point_line} associates the following two dots.
The dot $(0.5,1.5)$ means that the cluster $\{9,10\}$ merged at the current scale $s=1.5$ was previously formed at the smaller scale $s=0.5$.
The dot $(1,1.5)$ means that another cluster $\{4,6\}$ merged at the current scale $s=1.5$ was formed at $s=1$. 
\medskip

Every arc in the single-linkage dendrogram between nodes at scales $b$ and $d$ contributes one dot $(b,d)$ to the mergegram, e.g. both singleton sets $\{9\}$, $\{10\}$ merging at $s=0.5$ contribute two dots $(0,0.5)$ or one dot of multiplicity 2 shown in red, see Fig.~\ref{fig:5-point_line}. 
\end{exa}

Example~\ref{exa:5-point_line} shows that the mergegram $\MG$ retains more geometric information of a set $A$ than the persistence diagram $\PD$.
It turns out that this new intermediate object (larger than $\PD$ and smaller than a full dendrogram) enjoys the stability of persistence, which makes $\MG$ useful for analysing noisy data in all cases when distance-based 0D persistence is used.
\medskip

Here is the summary of new contributions to Topological Data Analysis.
\smallskip

\noindent
$\bullet$
Definition~\ref{dfn:mergegram} introduces the concept of a mergegram for any dendrogram of clustering.
\smallskip

\noindent
$\bullet$
Theorem~\ref{thm:0D_persistence_mergegram} and Example~\ref{exa:mergegram_stronger} justify that the mergegram of a single-linkage dendrogram is strictly stronger than the 0D persistence of a distance-based filtration of sublevel sets.
\smallskip

\noindent
$\bullet$
Theorem~\ref{thm:stability_mergegram} proves that the mergegram of any single-linkage dendrogram is stable in the bottleneck distance under  perturbations of a finite set in the Hausdorff distance.
\smallskip

\noindent
$\bullet$
Theorem~\ref{thm:complexity} shows that the mergegram can be computed in a near linear time.

\section{Related work on hierarchical clustering and deep neural networks}
\label{sec:review}

The aim of clustering is to split a given set of points into clusters such that points within one cluster are more similar to each other than points from different clusters.
\medskip

A clustering problem can be made exact by specifying a distance between given points and restrictions on outputs, e.g. a number of clusters or a cost function to minimize.
\medskip

All hierarchical clustering algorithms can output a hierarchy of clusters or a dendrogram visualising mergers of clusters as explained later in Definition~\ref{dfn:dendrogram}.
Here we introduce only the simplest single-linkage clustering, which plays the central role in the paper.

\begin{dfn}[single-linkage clustering]
\label{dfn:sl_clustering}
Let $A$ be a finite set in a metric space $X$ with a distance $d:X\times X\to[0,+\infty)$.
Given a distance threshold, which will be called a scale $s$, any points $a,b\in A$ should belong to one \emph{SL cluster} if and only if there is a finite sequence $a=a_1,\dots,a_m=b\in A$ such that any two successive points have a distance at most $s$, i.e. $d(a_i,a_{i+1})\leq s$ for $i=1,\dots,m-1$.
Let $\De_{SL}(A;s)$ denote the collection of SL clusters at the scale $s$.
For $s=0$, any point $a\in A$ forms a singleton cluster $\{a\}$.
Representing each cluster from $\De_{SL}(A;s)$ over all $s\geq 0$ by one point, we get the \emph{single-linkage dendrogram} $\De_{SL}(A)$ visualizing how clusters merge, see the first bottom picture in Fig.~\ref{exa:5-point_line}.
\bs
\end{dfn}

Another way to visualize SL clusters is to build a Minimum Spanning Tree below.

\begin{dfn}[Minimum Spanning Tree $\MST(A)$]
\label{dfn:mst}
The \emph{Minimum Spanning Tree} $\MST(A)$ of a finite set $A$ in a metric space $X$ with a distance $d$ is a tree (a connected graph without cycles) that has the vertex set $A$ and the minimum total length of edges.
We assume that the length of any edge between vertices $a,b\in A$ is measured as $d(a,b)$.
\bs
\end{dfn}

A review of the relevant past work on persistence diagrams is postponed to section~\ref{sec:persistence_modules}, which introduces more auxiliary notions.
A persistence diagram consists of dots $(b,d)\in\R^2$ whose birth/death coordinates represent a life interval $[b,d)$ of a homology class, e.g. a connected component in a Vietoris-Rips filtration, see the bottom middle picture in Fig.~\ref{fig:5-point_line}.
\medskip

Persistence diagrams are isometry invariants that are stable under noise in the sense that a topological space and its noisy point sample have close persistence diagrams.
This stability under noise allows us to classify continuous shapes by using only their discrete samples.
\medskip

Imagine that several rigid shapes are sparsely represented by a few salient points, e.g. corners or local maxima of a distance function.
Translations and rotations of these point clouds do not change the underlying shapes.
Hence clouds should be classified modulo isometries that preserve distances between points.
The important problem is to recognize of a shape, e.g. within a given set of representatives, from its sparse point sample with noise.
This paper solves the problem by computing isometry invariants, namely the new mergegram, the 0D persistence and the pair-set of distances to two nearest neighbors for each point.
\medskip

Since all dots in a persistence diagram are unordered, our experimental  section~\ref{sec:experiments} uses a neural network whose output is invariant under permutations of input point by construction.
PersLay \cite{carriere2019perslay} is a collection of permutation invariant neural network layers i.e. functions on sets of points in $\mathbb{R}^n$ that give the same output regardless of the order they are inserted.  
\medskip

PersLay extends the neural network layers introduced in Deep Sets \cite{zaheer2017deep}. 
Perslay introduces new layers to specially handle persistence diagrams, as well as new form of representing such layers. 
Each layer is a combination of a coefficient layer $\omega(p):\mathbb{R}^n \rightarrow \mathbb{R}$, point transformation $\phi(p):\mathbb{R}^n \rightarrow \mathbb{R}^q$ and permutation invariant layer $\text{op}$ to retrieve the final output
$$\text{PersLay}(\text{diagram}) = \text{op}(\{\omega(p)\phi(p)\}), \text{ where } p \in \text{diagram (any set of points in }  \R^n).$$

\section{The merge module and mergegram of a dendrogram}
\label{sec:mergegram}

The section introduces a merge module (a family of vector spaces with consistent linear maps) and a mergegram (a diagram of points in $\R^2$ representing a merge module).

\begin{dfn}[partition set $\PS(A)$]
\label{dfn:partition}
For any set $A$, a \emph{partition} of $A$ is a finite collection of non-empty disjoint subsets $A_1,\dots,A_k\subset A$ whose union is $A$.
The \emph{single-block} partition of $A$ consists of the set $A$ itself.
The \emph{partition set} $\PS(A)$ consists of all partitions of $A$.
\bs
\end{dfn}

If $A=\{1,2,3\}$, then $(\{1,2\},\{3\})$ is a partition of $A$, but
$(\{1\},\{2\})$ and $(\{1,2\},\{1,3\})$ are not.
In this case the partition set $\PS(A)$ consists of 5 partitions
$$(\{1\},\{2\},\{3\}),\quad
(\{1,2\},\{3\}),\quad
(\{1,3\},\{2\}),\quad
(\{2,3\},\{1\}),\quad
(\{1,2,3\}).$$

Definition~\ref{dfn:dendrogram} below extends the concept of a dendrogram from \cite[section~3.1]{carlsson2010characterization} to arbitrary (possibly, infinite) sets $A$.
Since every partition of $A$ is finite by Definition~\ref{dfn:partition}, we don't need to add that an initial partition of $A$ is finite.
Non-singleton sets are now allowed.

\begin{dfn}[dendrogram of merge sets]
\label{dfn:dendrogram}
A \emph{dendrogram} over any set $A$ is a function $\Delta:[0,\infty)\to\PS(A)$ of a scale $s\geq 0$ satisfying the following conditions.
\smallskip

\noindent
(\ref{dfn:dendrogram}a)
There exists a scale $r\geq 0$ such that $\De(A;s)$ is the single block partition for all $s\geq r$. 
\smallskip

\noindent
(\ref{dfn:dendrogram}b)
If $s\leq t$, then $\De(A;s)$ \emph{refines} $\De(A;t)$, i.e. any set from $\De(t)$ is a subset of some set from $\De(A;t)$.
These inclusions of subsets of $X$ induce the natural map $\De_s^t:\De(s)\to\De(t)$.
\smallskip

\noindent
(\ref{dfn:dendrogram}c)
There are finitely many \emph{merge scales} $s_i$ such that $$s_0 = 0 \text{ and  } s_{i+1} = \text{sup}\{s \mid \text{ the map }  \De_s^t \text{ is identity for } s' \in [s_i,s)\}, i=0,\dots,m-1.$$

\noindent
Since $\De(A;s_{i})\to\De(A;s_{i+1})$ is not an identity map, there is a subset $B\in\De(s_{i+1})$ whose preimage consists of at least two subsets from $\De(s_{i})$.
This subset $B\subset X$ is called a \emph{merge} set and its \emph{birth} scale is $s_i$.
All sets of $\De(A;0)$ are merge sets at the birth scale 0.
The $\life(B)$ is the interval $[s_i,t)$ from its birth scale $s_i$ to its \emph{death} scale $t=\sup\{s \mid \De_{s_i}^s(B)=B\}$.
\bs
\end{dfn}

Dendrograms are usually represented as trees whose nodes correspond to all sets from the partitions $\De(A;s_i)$ at merge scales.
Edges of such a tree connect any set $B\in\De(A;s_{i})$ with its preimages under $\De(A;s_{i})\to\De(A;s_{i+1})$.
Fig.~\ref{fig:3-point_dendrogram} shows the dendrogram on $A=\{1,2,3\}$.
\medskip

\begin{figure}
\parbox{100mm}{
\begin{tabular}{lccccc}
partition $\De(A;2)$ at scale $s_2=2$ & & & $\{1,2,3\}$ & & \\
map $\De_1^2:\De(A;1)\to\De(A;2)$ & & & $\uparrow$ & $\nwarrow$ & \\
partition $\De(A;1)$ at scale $s_1=1$ & & & \{1, 2\} & & \{3\} \\
map $\De_0^1:\De(A;0)\to\De(A;1)$ & & $\nearrow$ & $\uparrow$ & & $\uparrow$  \\
partition $\De(A;0)$ at scale $s_0=0$ & $\{1\}$ & & $\{2\}$ & & \{3\} 
\end{tabular}}
\parbox{35mm}{
\begin{tikzpicture}[scale = 0.9]
  \draw[style=help lines,step = 1] (0,0) grid (2.4,2.4);
  \draw[->] (-0.2,0) -- (2.4,0) node[right] {birth};
  \draw[->] (0,-0.2) -- (0,2.4) node[above] {death};	
  \draw[-] (0,0) -- (2.4,2.4) node[right]{};
  \foreach \x/\xtext in {1/1, 2/2}
    \draw[shift={(\x,0)}] (0pt,2pt) -- (0pt,-2pt) node[below] {$\xtext$};
  \foreach \y/\ytext in {1/1, 2/2}
    \draw[shift={(0,\y)}] (2pt,0pt) -- (-2pt,0pt) node[left] {$\ytext$};
   \filldraw [fill=red] (0,1) circle (2pt);
   \filldraw [fill = red] (0,1) circle (2pt);
   \filldraw [fill = blue] (0,2) circle (2pt);
   \filldraw [fill = blue] (1,2) circle (2pt);
   \filldraw [fill = blue] (2,2.7) circle (2pt);
\end{tikzpicture}}
\caption{The dendrogram $\De$ on $A=\{1,2,3\}$ and its mergegram  $\MG(\De)$ from Definition~\ref{dfn:mergegram}.}
\label{fig:3-point_dendrogram}
\end{figure}
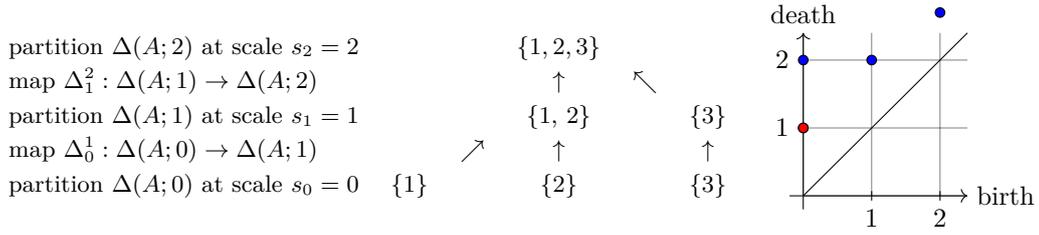

In the dendrogram above, the partition $\De(A;1)$ consists of $\{1,2\}$ and $\{3\}$.
The maps $\De_s^t$ induced by inclusions respect the compositions in the sense that $\De_s^t\circ\De_r^s=\De_r^t$ for any $r\leq s\leq t$, e.g. $\De_0^1(\{1\})=\{1,2\}=\De_0^1(\{2\})$ and $\De_0^1(\{3\})=\{3\}$, i.e. $\De_0^1$ is a well-defined map from the partition $\De(A;0)$ in 3 singleton sets to $\De(A;1)$, but isn't an identity.
\medskip

At the scale $s_0=0$ the merge sets $\{1\},\{2\}$ have $\life=[0,1)$, while the merge set $\{3\}$ has $\life=[0,2)$.
At the scale $s_1=1$ the only merge set $\{1,2\}$ has $\life=[1,2)$.
At the scale $s_2=2$ the only merge set $\{1,2,3\}$ has $\life=[2,+\infty)$.
The notation $\De$ is motivated as the first (Greek) letter in the word dendrogram and by a $\De$-shape of a typical tree above.
\medskip

Condition~(\ref{dfn:dendrogram}a) means that 
a partition of $X$ is trivial for all large scales $s$.
Condition~(\ref{dfn:dendrogram}b) says that when the scale $s$ in increasing sets from a partition $\De(s)$ can only merge with each other, but can not split. 
Condition~(\ref{dfn:dendrogram}c) implies that there are only finitely many mergers, when two or more subsets of $X$ merge into a larger merge set.
\medskip

\begin{lem}[single-linkage dendrogram]
\label{lem:sl_clustering}
Given a metric space $(X,d)$ and a finite set $A\subset X$, the single-linkage dendrogram $\De_{SL}(X)$ from Definition~\ref{dfn:sl_clustering} satisfies Definition~\ref{dfn:dendrogram}.
\end{lem}
\begin{proof}
Since $A$ is finite, there are only finitely many inter-point distances within $A$, which implies condition (\ref{dfn:dendrogram}a,c).
Let $f(p):X\to\R$ be the distance from a point $p\in X$ to (the closest point of) $A$.
Condition (\ref{dfn:dendrogram}b) follows  the inclusions $f^{-1}[0,s) \subseteq f^{-1}[0,t) $ for $s\leq t$. 
\end{proof}

A \emph{mergegram} represents lives of merge sets by dots with two coordinates (birth,death).

\begin{dfn}[mergegram $\MG(\De)$]
\label{dfn:mergegram}
The \emph{mergegram} of a dendrogram $\De$ from Definition~\ref{dfn:dendrogram} has the dot (birth,death) in $\R^2$ for each merge set $A$ of $\De$ with $\life(A)$=[birth,death).
If any life interval appears $k$ times, the dot (birth,death) has the multiplicity $k$ in $\MG(\De)$.
\bs
\end{dfn}

For simplicity, this paper considers vector spaces with coefficients (of linear combinations of vectors) only in $\Z_2=\{0,1\}$, which can be replaced by any field.

\begin{dfn}[merge module $M(\De)$]
\label{dfn:merge_module}
For any dendrogam $\De$ on a set $X$ from Definition~\ref{dfn:dendrogram},
the \emph{merge module} $M(\De)$ consists of the vector spaces $M_s(\De)$, $s\in\R$, and linear maps $m_s^t:M_s(\De)\to M_t(\De)$, $s\leq t$.
For any $s\in\R$ and $A\in\De(s)$, the space $M_s(\De)$ has the generator or a basis vector $[A]\in M_s(\De)$.
For $s<t$ and any set $A\in\De(s)$, 
if the image of $A$ under $\De_s^t$ coincides with $A\subset X$, i.e. $\De_s^t(A)=A$, then $m_s^t([A])=[A]$, else $m_s^t([A])=0$. 
\bs
\end{dfn}

\begin{figure}[h]
\begin{tabular}{lccccccccc}
scale $s_3=+\infty$ & 0 & & & & & 0 \\
map $m_2^{+\infty}$ & $\uparrow$ & & & & & $\uparrow$\\
scale $s_2=2$ & $\Z_2$ & & & 0 & 0 & [\{1,2,3\}]\\
map $m_1^2$ & $\uparrow$ & & & $\uparrow$ & $\uparrow$\\
scale $s_1=1$ & $\Z_2\oplus\Z_2$ & 0 & 0 & [\{3\}] & [\{1,2\}] \\
map $m_0^1$ & $\uparrow$ & $\uparrow$ & $\uparrow$ & $\uparrow$ \\
scale $s_0=0$ & $\Z_2\oplus\Z_2\oplus\Z_2$ & [\{1\}] & [\{2\}] & [\{3\}] &
\end{tabular}
\caption{The merge module $M(\De)$ of the dendrogram $\De$ on the set $X=\{1,2,3\}$ in Fig.~\ref{fig:3-point_dendrogram}.}
\label{fig:3-point_module}
\end{figure}

\begin{exa}
\label{exa:5-point_set}
Fig.~\ref{fig:5-point_set} shows the metric space $X=\{a,b,c,d,e\}$ with distances defined by the shortest path metric induced by the specified edge-lengths, see the distance matrix. 

\begin{figure}[H]
\parbox{80mm}{
\begin{tikzpicture}[scale = 0.75][sloped]
\node (x) at (5,3) {x};
 \node (a) at (1,1) {a};
 \draw (a) -- node[above]{5} ++ (x);
 \node (b) at (3.5,4.0) {b};
 \draw (b) -- node[above]{1} ++ (x);
 \node (c) at (7,1) {c};
 \draw (c) -- node[below]{2} ++ (x);
 \node (y) at (8,3) {y};
 \draw (x) -- node[above]{2} ++ (y);
 \node (d) at (10,5){p};
 \node (e) at (10,1){q};
 \draw (y) -- node[below]{2} ++ (d);
 \draw (y) -- node[below]{2} ++ (e);
\end{tikzpicture}}
\parbox{40mm}{
\begin{tabular}{c|ccccc}
& a & b & c & p & q \\
\hline
a & 0 & 6 & 7 & 9 & 9 \\
b & 6 & 0 & 3 & 5 & 5 \\
c & 7 & 3 & 0 & 6 & 6 \\
p & 9 & 5 & 6 & 0 & 4 \\
q & 9 & 5 & 6 & 4 & 0 
\end{tabular}}
\caption{The set $X=\{a,b,c,d,e\}$ has the distance matrix defined by the shortest path metric.}
\label{fig:5-point_set}
\end{figure}
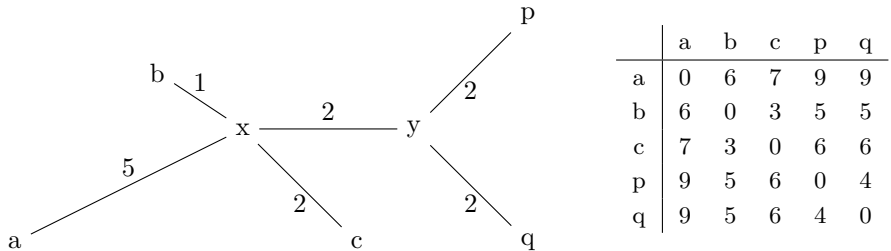

\begin{figure}[H]
\begin{tikzpicture}[scale = 0.6][sloped]
 \draw[style=help lines,step = 1] (-1,0) grid (10.4,6.3);

\foreach \i in {0,0.5,...,3.0} { \node at (-1.4,2*\i) {\i}; }
\node (a) at (0,-0.3) {a};
\node (b) at (4,-0.3) {b};
\node (c) at (6,-0.3) {c};
\node (d) at (8,-0.3) {p};
\node (e) at (10,-0.3) {q};
\node (x) at (5,6.75) {};

\node (de) at (9,4){};
\node (bc) at (5.0,3){};
\node (bcde) at (7.0,5){};
\node (all) at (5.0,6){};

\draw [line width=0.5mm, blue ] (a) |- (all.center);
\draw [line width=0.5mm, blue ] (b) |- (bc.center);
\draw [line width=0.5mm, blue ]  (c) |- (bc.center);
\draw [line width=0.5mm, blue ]  (d) |- (de.center);
\draw [line width=0.5mm, blue ]  (e) |- (de.center);
\draw [line width=0.5mm, blue ]  (de.center) |- (bcde.center);
\draw [line width=0.5mm, blue ]  (bc.center) |- (bcde.center);
\draw [line width=0.5mm, blue ]  (bcde.center) |- (all.center);
\draw [line width=0.5mm, blue] [->] (all.center) -> (x.center);
\end{tikzpicture}
\hspace*{1cm}
\begin{tikzpicture}[scale = 1.1]
 \draw[style=help lines,step = 0.5] (0,0) grid (3.4,3.4);
 \draw[->] (-0.2,0) -- (3.4,0) node[right] {birth};
  \draw[->] (0,-0.2) -- (0,3.4) node[above] {death};	
  \draw[-] (0,0) -- (3.4,3.4) node[right]{};

  \foreach \x/\xtext in {0.5/0.5, 1/1, 1.5/1.5, 2.0/2, 2.5/2.5, 3.0/3}
    \draw[shift={(\x,0)}] (0pt,2pt) -- (0pt,-2pt) node[below] {$\xtext$};

  \foreach \y/\ytext in {0.5/0.5,  1/1, 1.5/1.5, 2.0/2, 2.5/2.5, 3.0/3}
    \draw[shift={(0,\y)}] (2pt,0pt) -- (-2pt,0pt) node[left] {$\ytext$}; 
   \filldraw[fill=red] (0,1.5) circle (2pt);
   \filldraw [fill = red] (0.0,2.0) circle (2pt);
   \filldraw [fill = blue] (0.0,3) circle (2pt);
   \filldraw [fill = blue] (1.5,2.5) circle (2pt);
   \filldraw [fill = blue] (2.0,2.5) circle (2pt);
   \filldraw [fill = blue] (2.5, 3.0) circle (2pt);
   \filldraw [fill = blue] (3, 3.7) circle (2pt);
\end{tikzpicture}
\caption{\textbf{Left}: the dendrogram $\De$ for the single linkage clustering of the set 5-point set $X=\{a,b,c,d,e\}$ in Fig.~\ref{fig:5-point_set}.
\textbf{Right}: the mergegram $\MG(\De)$, red dots have multiplicity 2.}
\label{fig:5-point_set_mergegram}
\end{figure}
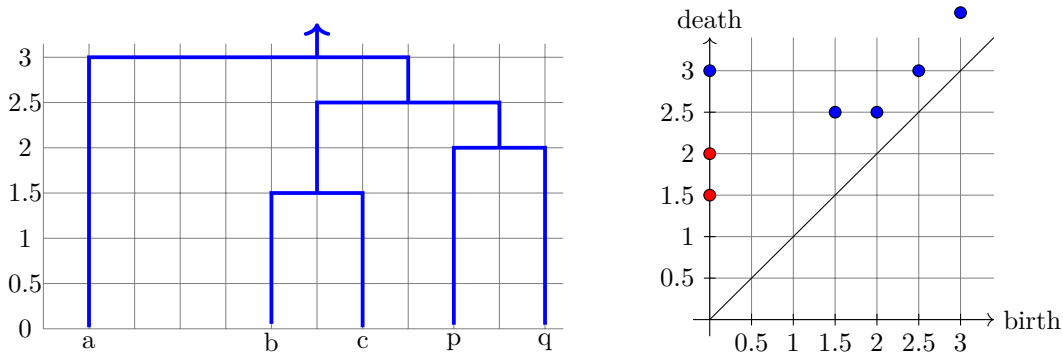

The dendrogram $\De$ in the first picture of Fig.~\ref{fig:5-point_set_mergegram} generates the mergegram as follows:
\begin{itemize}
\item 
each of the singleton sets $\{b\}$ and $\{c\}$ has the dot (0,1.5), so its multiplicity is 2; 
\item 
each of the singleton sets $\{p\}$ and $\{q\}$ has the dot (0,2), so its multiplicity is 2; 
\item 
the singleton set $\{a\}$ has the dot $(0,3)$;
the merge set $\{b,c\}$ has the dot (1.5,2.5);
\item
the merge set $\{p,q\}$ has the dot (2,2.5);
the merge set $\{b,c,p,q\}$ has the dot (2.5,3);
\item
the merge set $\{a,b,c,p,q\}$ has the dot $(3,+\infty)$.
\end{itemize}
\end{exa}

\section{Background on persistence modules and diagrams}
\label{sec:persistence_modules}

This section introduces the key concepts from the thorough review by Chazal et al. \cite{chazal2016structure}.
As will become clear soon, the merge module of any dendrogram belongs to a wider class below.

\begin{dfn}[persistence module $\V$]
\label{dfn:persistence_module}
A \emph{persistence module} $\mathbb{V}$ over the real numbers $\mathbb{R}$ is a family of vector spaces $V_t$, $t\in \mathbb{R}$ with linear maps $v^t_s:V_s \rightarrow V_t$, $s\leq t$ such that $v^t_t$ is the identity map on $V_t$ and the composition is respected: $v^t_s \circ v^s_r = v^t_r$ for any $r \leq s \leq t$.
\bs
\end{dfn}

The set of real numbers can be considered as a category  $\mathbb{R}$ in the following sense.
The objects of $\R$ are all real numbers. 
Any two real numbers such that $a\leq b$ define a single morphism $a\to b$.
The composition of morphisms $a\to b$ and $b \to c$ is the morphism $a \leq c$. 
In this language, a persistence module is a functor from $\R$ to the category of vector spaces.
\medskip

A basic example of $\V$ is an interval module.
An interval $J$ between points $p<q$ in the line $\R$ can be one of the following types: closed $[p,q]$, open $(p,q)$ and half-open or half-closed $[p,q)$ and $(p,q]$.
It is convenient to encode types of endpoints by $\pm$ superscripts as follows:
$$[p^-,q^+]:=[p,q],\quad
[p^+,q^-]:=(p,q),\quad
[p^+,q^+]:=(p,q],\quad
[p^-,q^-]:=[p,q).$$

The endpoints $p,q$ can also take the infinite values $\pm\infty$, but without superscripts.

\begin{exa}[interval module $\I(J)$]
\label{exa:interval_module}
For any interval $J\subset\R$, the \emph{interval module} $\I(J)$ is the persistence module defined by the following vector spaces $I_s$ and linear maps $i_s^t:I_s\to I_t$
$$I_s=\left\{ \begin{array}{ll} 
\Z_2, & \mbox{ for } s\in J, \\
0, & \mbox{ otherwise }; 
\end{array} \right.\qquad
i_s^t=\left\{ \begin{array}{ll} 
\id, & \mbox{ for } s,t\in J, \\
0, & \mbox{ otherwise }
\end{array} \right.\mbox{ for any }s\leq t.$$
\end{exa}
\medskip

The direct sum $\W=\U\oplus\V$ of persistence modules $\U,\V$ is defined  as the persistence module with the vector spaces $W_s=U_s\oplus V_s$ and linear maps $w_s^t=u_s^t\oplus v_s^t$.
\medskip

We illustrate the abstract concepts above using geometric constructions of Topological Data Analysis.
Let $f:X\to\R$ be a continuous function on a topological space.
Its \emph{sublevel} sets $X_s^f=f^{-1}((-\infty,s])$ form nested subspaces $X_s^f\subset X_t^f$ for any $s\leq t$.
The inclusions of the sublevel sets respect compositions similarly to a dendrogram $\De$ in Definition~\ref{dfn:dendrogram}.
\medskip

On a metric space $X$ with with a distance function $d:X\times X\to[0,+\infty)$, a typical example of a function $f:X\to\R$ is the distance to a finite set of points $A\subset X$. 
More specifically, for any point $p\in X$, let $f(p)$ be the distance from $p$ to (a closest point of) $A$.
For any $r\geq 0$, the preimage $X_r^f=f^{-1}((-\infty,r])=\{q\in X \mid d(q,A)\leq r\}$ is the union of closed balls that have the radius $r$ and centers at all points $p\in A$.
For example, $X_0^f=f^{-1}((-\infty,0])=A$ and $X_{+\infty}^f=f^{-1}(\R)=X$.
\medskip

If we consider any continuous function $f:X\to\R$, we have the inclusion $X_s^f\subset X_r^f$ for any $s\leq r$.
Hence all sublevel sets $X_s^f$ form a nested sequence of subspaces within $X$.
The above construction of a \emph{filtration} $\{X_s^f\}$ can be considered as a functor from $\R$ to the category of topological spaces.  
Below we discuss the most practically used case of dimension 0.

\begin{exa}[persistent homology]
\label{exa:persistent_homology}
For any topological space $X$,  the 0-dimensional \emph{homology} $H_0(X)$ is the vector space (with coefficients $\Z_2$) generated by all connected components of $X$.
Let $\{X_s\}$ be any \emph{filtration} of nested spaces, e.g. sublevel sets $X_s^f$ based on a continuous function $f:X\to\R$.
The inclusions $X_s\subset X_r$ for $s\leq r$ induce the linear maps between homology groups $H_0(X_s)\to H_0(X_r)$ and define the \emph{persistent homology} $\{H_0(X_s)\}$, which satisfies the conditions of a persistence module from Definition~\ref{dfn:persistence_module}.
\bs
\end{exa}
\medskip

If $X$ is a finite set of $m$ points, then $H_0(X)$ is the direct sum $\Z_2^m$ of $m$ copies of $\Z_2$.  
\medskip

The persistence modules that can be decomposed as direct sums of interval modules can be described in a very simple combinatorial way by persistence diagrams of dots in $\R^2$.

\begin{dfn}[persistence diagram $\PD(\V)$]
\label{dfn:persistence_diagram}
Let a persistence module $\V$ be decomposed as a direct sum of interval modules from Example~\ref{exa:interval_module} : $\V\cong\bigoplus\limits_{l \in L}\I(p^{*}_l,q^{*}_l)$, where $*$ is $+$ or $-$.
The \emph{persistence diagram} $\PD(\V)$ is the multiset 
$\PD(\mathbb{V}) = \{(p_l,q_l) \mid l \in L \} \setminus \{p=q\}\subset\R^2$.
\bs
\end{dfn}
\medskip

The 0-dimensional persistent homology of a space $X$ with a continuous function $f:X\to\R$ will be denoted by $\PD\{H_0(X_s^f)\}$.
Lemma~\ref{lem:merge_module_decomposition} will prove that the merge module $M(\De)$ of any dendrogram $\De$ is also decomposable into interval modules.
Hence the mergegram $\MG(\De)$ from Definition~\ref{dfn:mergegram} can be interpreted as the persistence diagram of the merge module $M(\De)$.   

\section{The mergegram is stronger than the 0-dimensional persistence}
\label{sec:mergegram_stronger}

Let $f:X\to\R$ be the distance function to a finite subset $A$ of a metric space $(X,d)$.
The persistent homology $\{H_k(X_s^f)\}$ in any dimension $k$ is invariant under isometries of $X$.
\medskip

Moreover, the persistence diagrams of very different shapes, e.g. topological spaces and their discrete samples, can be easily compared by the bottleneck distance in Definition~\ref{dfn:bottleneck_distance}.
\medskip

Practical applications of persistence are justified by Stability Theorem~\ref{thm:stability_persistence} saying that the persistence diagram continuously changes under perturbations of a given filtration or an initial point set.
A similar stability of mergegrams will be proved in Theorem~\ref{thm:stability_mergegram}.
\medskip

This section shows that the mergegram $\MG(\De_{SL}(A))$ has more isometry information about the subset $A\subset X$ than the 0-dimensional persistent homology $\{H_0(X_s^f)\}$.
\medskip

Theorem~\ref{thm:0D_persistence_mergegram} shows how to obtain the 0D persistence $\PD\{H_0(X_s^f)\}$ from $\MG(\De_{SL}(A))$, where $f:X\to\R$ is the distance to a finite subset $A\subset X$.
Example~\ref{exa:mergegram_stronger} builds two 4-point sets in $\R$ whose persistence diagrams are identical, but their mergegrams are different.
\medskip

We start from folklore Claims~\ref{claim:0D_persistence_SL}-\ref{claim:0D_persistence_MST}, which interpret the 0D persistence $\PD\{H_0(X_s^f)\}$ using the classical concepts of the single-linkage dendrogram and Minimum Spanning Tree.

\begin{myclaim}[0D persistence from $\De_{SL}$]
\label{claim:0D_persistence_SL}
For a finite set $A$ in a metric space $(X,d)$, let $f:X\to\R$ be the distance to $A$.
In the single-linkage dendrogram $\De_{SL}(A)$, let $0<s_1<\dots<s_m<s_{m+1}=+\infty$ be all distinct merge scales.
If $k\geq 2$ subsets of $A$ merge into a larger subset of $A$ at a scale $s_i$, the multiplicity of $s_i$ is $\mu_i=k-1$.
Then the persistence diagram $\PD\{H_0(X_s^f)\}$ consists of the dots $(0,s_i)$ with multiplicities $\mu_i$, $i=1,\dots,m+1$.
\bs
\end{myclaim}

\begin{myclaim}[0D persistence from MST]
\label{claim:0D_persistence_MST}
For a set $A$ of $n$ points in a metric space $(X,d)$, let $f:X\to\R$ be the distance to $A$.
Let a Minimum Spanning Tree $\MST(A)$ have edge-lengths $l_1\leq\dots\leq l_{n-1}$.
The persistence diagram $\PD\{H_0(X_s^f)\}$ consists of the $n-1$ dots $(0,0.5l_i)$ counted with multiplicities if some edge-lengths are equal, plus the infinite dot $(0,+\infty)$.
\bs
\end{myclaim}

\begin{thm}[0D persistence from a mergegram]
\label{thm:0D_persistence_mergegram} 
For a finite set $A$ in a metric space $(X,d)$, let $f:X\to\R$ be the distance to $A$.
Let the mergegram $\MG(\De_{SL}(A))$ be a multiset $\{(b_i,d_i)\}_{i=1}^k$, where some dots can be repeated.
Then the persistence diagram $\PD\{H_0(X_s^f)\}$ is the difference of the multisets $\{(0,d_i)\}_{i=1}^{k}-\{(0,b_i)\}_{i=1}^{k}$ containing each dot $(0,s)$ exactly $\#b-\#d$ times, where $\#b$ is the number of births $b_i=s$, $\#d$ is the number of deaths $d_i=s$.
All trivial dots $(0,0)$ are ignored, alternatively we take $\{(0,d_i)\}_{i=1}^{k}$ only with $d_i>0$.
\bs
\end{thm}
\begin{proof}
In the language of Claim~\ref{claim:0D_persistence_SL}, let at a scale $s>0$ of multiplicity $\mu$ exactly $\mu+1$ subsets merge into a set $B\in\De_{SL}(A;s)$.
By Claim~\ref{claim:0D_persistence_SL} this set $B$ contributes $\mu$ dots $(0,s)$ to the persistence diagrams $\PD\{H_0(X_s^f)\}$.
By Definition~\ref{dfn:mergegram} the same set $B$ contributes $\mu+1$ dots of the form $(b_i,s)$, $i=1,\dots,\mu+1$, corresponding to the $\mu+1$ sets that merge into $B$ at the scale $s$.
Moreover, the set $B$ itself will merge later into a larger set, which creates one extra dot $(s,d)\in\PD\{H_0(X_s^f)\}$.
The exceptional case $B=A$ corresponds to $d=+\infty$.
\smallskip

If we remove one dot $(0,s)$ from the $\mu+1$ dots counted above
 as expected in the difference $\{(0,d_i)\}_{i=1}^{k}-\{(0,b_i)\}_{i=1}^{k}$ of multisets, we get exactly $\mu$ dots $(0,s)\in\PD\{H_0(X_s^f)\}$.
The required formula has been proved for contributions of any merge set $B\subset A$.
\end{proof}

In Example~\ref{exa:5-point_line} the mergegram in the last picture of Fig.~\ref{fig:5-point_line} is the multiset of 9 dots:
$$\MG(\De_{SL}(A))=\{(0,0.5),(0,0.5),(0,1),(0,1),(0.5,1.5),(1,1.5),(0,2),(1.5,2),(2,+\infty)\}.$$
Taking the difference of multisets and ignoring trivial dots $(0,0)$, we get \\
$\PD(H_0\{X_s^f\})=\{(0,0.5),(0,0.5),(0,1),(0,1),(0,1.5),(0,1.5),(0,2),(0,2),(0,+\infty)\}-$ \\
$-\{(0,0.5),(0,1),(0,2)\}=\{(0,0.5),(0,1),(0,1.5),(0,2),(0,+\infty)\}
\mbox{ as in Fig.~\ref{fig:5-point_line}}.$

\begin{exa}[the mergegram is stronger than 0D persistence]
\label{exa:mergegram_stronger}
Fig.~\ref{fig:4-point_set1} and~\ref{fig:4-point_set2} show the dendrograms, identical 0D persistence diagrams and different mergegrams for the sets $A=\{0,1,3,7\}$ and $B=\{0,1,5,7\}$ in $\R$.
This example together with Theorem~\ref{thm:0D_persistence_mergegram} justify that the new mergregram is strictly stronger than 0D persistence as an isometry invariant.

\begin{figure}[H]
\centering
\begin{tikzpicture}[scale = 0.5][sloped]
\draw [->] (-1,0) -- (-1,5) node[above] {scale $s$};
\draw[style=help lines,step = 1] (-1,0) grid (6.4,4.4);
\foreach \i in {0,0.5,...,2} { \node at (-1.5,2*\i) {\i}; }
\node (a) at (0,-0.4) {0};
\node (b) at (2,-0.4) {1};
\node (c) at (4,-0.4) {3};
\node (d) at (6,-0.4) {7};
\node (x) at (4.625,5) {};
\node (y) at (4.625,4) {};
\node (ab) at (1,1) {};
\node (abc) at (2.5,2){};
\node (abce) at (3,4){};
\draw [line width=0.5mm, blue] (a) |- (ab.center);
\draw [line width=0.5mm, blue] (b) |- (ab.center);
\draw [line width=0.5mm, blue] (ab.center) |- (abc.center);
\draw [line width=0.5mm, blue] (c) |- (abc.center);
\draw [line width=0.5mm, blue] (d) |- (abce.center);
\draw [line width=0.5mm, blue] (abc.center) |- (abce.center);
\draw [line width=0.5mm, blue] [->] (y.center) -> (x.center);
\end{tikzpicture}
\hspace*{2mm}
\begin{tikzpicture}[scale = 1.0]
  \draw[style=help lines,step = 0.5] (0,0) grid (2.4,2.4);
  \draw[->] (-0.2,0) -- (2.4,0) node[right] {birth};
  \draw[->] (0,-0.2) -- (0,2.4) node[above] {};	
  \draw[-] (0,0) -- (2.4,2.4) node[right]{};
  \foreach \x/\xtext in {0.5/0.5, 1/1, 1.5/1.5, 2.0/2}
    \draw[shift={(\x,0)}] (0pt,2pt) -- (0pt,-2pt) node[below] {$\xtext$};
  \foreach \y/\ytext in {0.5/0.5,  1/1, 1.5/1.5, 2.0/2}
    \draw[shift={(0,\y)}] (2pt,0pt) -- (-2pt,0pt) node[left] {$\ytext$};
   \filldraw [fill=blue] (0,0.5) circle (2pt);
   \filldraw [fill=blue] (0,1) circle (2pt);
   \filldraw [fill=blue] (0,2) circle (2pt);
   \filldraw [fill=blue] (0,2.6) circle (2pt);
\end{tikzpicture}
\hspace*{2mm}
\begin{tikzpicture}[scale = 1.0]
  \draw[style=help lines,step = 0.5] (0,0) grid (2.4,2.4);
  \draw[->] (-0.2,0) -- (2.4,0) node[right] {birth};
  \draw[->] (0,-0.2) -- (0,2.4) node[above] {death};	
  \draw[-] (0,0) -- (2.4,2.4) node[right]{};
  \foreach \x/\xtext in {0.5/0.5, 1/1, 1.5/1.5, 2.0/2}
    \draw[shift={(\x,0)}] (0pt,2pt) -- (0pt,-2pt) node[below] {$\xtext$};
  \foreach \y/\ytext in {0.5/0.5,  1/1, 1.5/1.5, 2.0/2}
    \draw[shift={(0,\y)}] (2pt,0pt) -- (-2pt,0pt) node[left] {$\ytext$};    
   \filldraw [fill=red] (0,0.5) circle (2pt);
   \filldraw [fill=blue] (0.0,1) circle (2pt);
   \filldraw [fill=blue] (0.0,2) circle (2pt);
   \filldraw [fill=blue] (0.5,1.0) circle (2pt);
   \filldraw [fill=blue] (1.0,2.0) circle (2pt);
   \filldraw [fill=blue] (2,2.6) circle (2pt);
\end{tikzpicture}
\caption{\textbf{Left}: single-linkage dendrogram \footnote{Note: The horizontal axes of the dendrograms are distorted} $\De_{SL}(A)$ for $A=\{0,1,3,7\}\subset\R$.
\textbf{Middle}: the 0D persistence diagram for the sublevel filtration of the distance to $A$.
\textbf{Right}: mergegram $\MG(\De_{SL}(A))$. }
\label{fig:4-point_set1}
\end{figure}
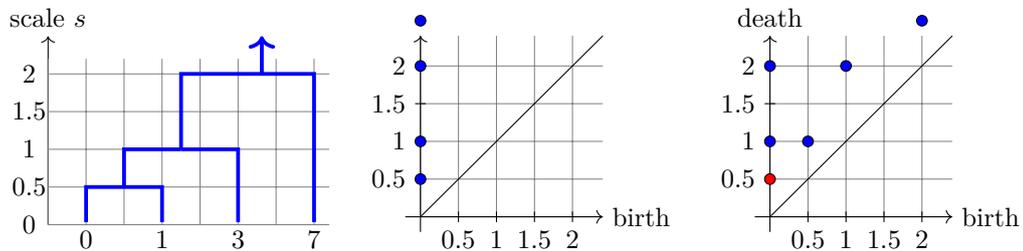

\begin{figure}[H]
\begin{tikzpicture}[scale = 0.5][sloped]
\draw [->] (-1,0) -- (-1,5) node[above] {scale $s$};
 \draw[style=help lines,step = 1] (-1,0) grid (6.4,4.4);
\foreach \i in {0,0.5,...,2} { \node at (-1.5,2*\i) {\i}; }
\node (a) at (0,-0.4) {0};
\node (b) at (2,-0.4) {1};
\node (c) at (4,-0.4) {5};
\node (d) at (6,-0.4) {7};
\node (ab) at (1,1) {};
\node (cd) at (5,2) {};
\node (abcd) at (3,4){};
\node (x) at (3,5) {};
\node (y) at (3,4) {};
\draw [line width=0.5mm, blue] (a) |- (ab.center);
\draw [line width=0.5mm, blue] (b) |- (ab.center);
\draw [line width=0.5mm, blue] (c) |- (cd.center);
\draw [line width=0.5mm, blue] (d) |- (cd.center);
\draw [line width=0.5mm, blue] (ab.center) |- (abcd.center);
\draw [line width=0.5mm, blue] (cd.center) |- (abcd.center);
\draw [line width=0.5mm, blue] [->] (y.center) -> (x.center);
\end{tikzpicture}
\hspace*{2mm}
\begin{tikzpicture}[scale = 1.0]
  \draw[style=help lines,step = 0.5] (0,0) grid (2.4,2.4);
  \draw[->] (-0.2,0) -- (2.4,0) node[right] {birth};
  \draw[->] (0,-0.2) -- (0,2.4) node[above] {};	
  \draw[-] (0,0) -- (2.4,2.4) node[right]{};
  \foreach \x/\xtext in {0.5/0.5, 1/1, 1.5/1.5, 2.0/2}
    \draw[shift={(\x,0)}] (0pt,2pt) -- (0pt,-2pt) node[below] {$\xtext$};
  \foreach \y/\ytext in {0.5/0.5,  1/1, 1.5/1.5, 2.0/2}
    \draw[shift={(0,\y)}] (2pt,0pt) -- (-2pt,0pt) node[left] {$\ytext$};    
   \filldraw [fill=blue] (0,0.5) circle (2pt);
   \filldraw [fill=blue] (0,1) circle (2pt);
   \filldraw [fill=blue] (0,2) circle (2pt);
   \filldraw [fill=blue] (0,2.6) circle (2pt);
\end{tikzpicture}
\hspace*{2mm}
\begin{tikzpicture}[scale = 1.0]
  \draw[style=help lines,step = 0.5] (0,0) grid (2.4,2.4);
  \draw[->] (-0.2,0) -- (2.4,0) node[right] {birth};
  \draw[->] (0,-0.2) -- (0,2.4) node[above] {death};	
  \draw[-] (0,0) -- (2.4,2.4) node[right]{};
  \foreach \x/\xtext in {0.5/0.5, 1/1, 1.5/1.5, 2.0/2}
    \draw[shift={(\x,0)}] (0pt,2pt) -- (0pt,-2pt) node[below] {$\xtext$};
  \foreach \y/\ytext in {0.5/0.5,  1/1, 1.5/1.5, 2.0/2}
    \draw[shift={(0,\y)}] (2pt,0pt) -- (-2pt,0pt) node[left] {$\ytext$};
   \filldraw[fill=red] (0,0.5) circle (2pt);
   \filldraw[fill = red] (0.0,1) circle (2pt);
   \filldraw [fill=blue] (0.5,2.0) circle (2pt);
   \filldraw [fill=blue] (1.0,2.0) circle (2pt);
   \filldraw [fill=blue] (2,2.6) circle (2pt);
\end{tikzpicture}
\caption{\textbf{Left}: single-linkage dendrogram $\De_{SL}(B)$ for $B=\{0,1,5,7\}\subset\R$.
\textbf{Middle}: the 0D persistence diagram for the sublevel filtration of the distance to $B$.
\textbf{Right}: mergegram $\MG(\De_{SL}(B))$. }
\label{fig:4-point_set2}
\end{figure}
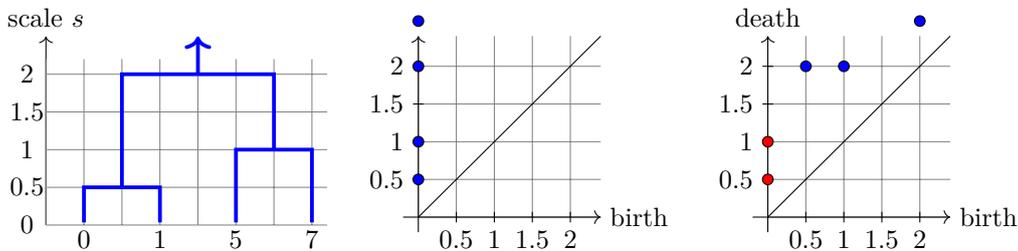
\end{exa}

\section{Distances and stability of persistence modules}
\label{sec:stability_persistence}

Definition~\ref{dfn:homo_modules} introduces homomorphisms between persistence modules, which are needed to state the stability of persistence diagrams $\PD\{H_0(X_s^f)\}$ under perturbations of a function $f:X\to\R$.
This result will imply a similar stability for the mergegram $\MG(\De_{SL}(A))$ for the dendrogram $\De_{SL}(A)$ of the single-linkage clustering of a set $A$ within a metric space $X$.

\begin{dfn}[a homomorphism of a degree $\de$ between persistence modules]
\label{dfn:homo_modules}
Let $\mathbb{U}$ and $\mathbb{V}$ be persistent modules over $\mathbb{R}$. 
A \emph{homomorphism} $\U\to\V$ of \emph{degree} $\delta\in\R$ is a collection of linear maps $\phi_t:U_t \rightarrow V_{t+\delta}$, $t \in \mathbb{R}$, such that the diagram commutes for all $s \leq t$. 
\begin{figure}[H]
\centering
\begin{tikzpicture}[scale=1.0]
  \matrix (m) [matrix of math nodes,row sep=3em,column sep=4em,minimum width=2em]
  {
     U_s & U_t \\
     V_{s+\delta} & V_{t+\delta} \\};
  \path[-stealth]
    (m-1-1) edge node [left] {$\phi_s$} (m-2-1)
            edge [-] node [above] {$u^t_s$} (m-1-2)
    (m-2-1.east|-m-2-2) edge node [above] {$v^{t+\delta}_{s+\delta}$}
            node [above] {} (m-2-2)
      (m-1-2) edge node [right] {$\phi_t$} (m-2-2);
\end{tikzpicture}
\end{figure}
Let $\text{Hom}^\delta(\mathbb{U},\mathbb{V})$ be all homomorphisms $\mathbb{U}\rightarrow \mathbb{V}$  of degree $\delta$.
Persistence modules $\U,\V$ are \emph{isomorphic} if they have inverse homomorphisms $\U\to\V$ and $\V\to\U$ of degree $\de=0$.
\bs
\end{dfn}

For a persistence module $\V$ with maps $v_s^t:V_s\to V_t$, the simplest example of a homomorphism of a degree $\de\geq 0$
 is $1_{\V}^{\de}:\V\to\V$ defined by the maps $v_s^{s+\de}$, $t\in\R$.
So the maps $v_s^t$ defining the structure of $\V$ shift all vector spaces $V_s$ the difference of scale $\de=t-s$.
\medskip

The concept of interleaved modules below is an algebraic generalization of a geometric perturbation of a set $X$ in terms of (the homology of) its sublevel sets $X_s$.

\begin{dfn}[interleaving distance ID]
\label{dfn:interleaving_distance}
Persistent modules $\mathbb{U}$ and $\mathbb{V}$ are $\delta$-interleaved if there are homomorphisms $\phi\in \text{Hom}^\delta(\mathbb{U},\mathbb{V})$ and $\psi \in \text{Hom}^\delta(\mathbb{V},\mathbb{U}) $ such that $\phi\circ\psi = 1_{\mathbb{V}}^{2\de} \text{ and } \psi\circ\phi = 1_{\mathbb{U}}^{2\de}$.
The \emph{interleaving distance} is 
$\ID(\U,\V)=\inf\{\de\geq 0 \mid \mathbb{U} \text{ and } \mathbb{V} \text{ are } \delta\text{-interleaved} \}$.
\bs
\end{dfn}

If $f,g:X\to\R$ are continuous functions such that $||f-g||_{\infty}\leq\de$ in the $L_{\infty}$-distance, the persistence modules $H_k\{f^{-1}(-\infty,s]\}$, $H_k\{g^{-1}(-\infty,s]\}$ are $\de$-interleaved for any $k$ \cite{cohen2007stability}.
The last conclusion extended to persistence diagrams in terms of the bottleneck distance below.

\begin{dfn}[bottleneck distance BD]
\label{dfn:bottleneck_distance}
Let multisets $C,D$ contain finitely many points $(p,q)\in\R^2$, $p<q$, of finite multiplicity and all diagonal points $(p,p)\in\R^2$ of infinite multiplicity.
For $\de\geq 0$, a $\de$-matching is a bijection $h:C\to D$ such that $|h(a)-a|_{\infty}\leq\de$ in the $L_{\infty}$-distance on the plane for any point $a\in C$.
The \emph{bottleneck} distance between persistence modules $\U,\V$ is $\BD(\mathbb{U},\mathbb{V}) = \text{inf}\{ \delta \mid \text{ there is a }\delta\text{-matching between } \PD(\mathbb{U}) \text{ and } \PD(\mathbb{V})\}$. 
\bs
\end{dfn}

The original stability of persistence for sequences of sublevel sets persistence was extended as Theorem~\ref{thm:stability_persistence} to $q$-tame persistence modules. 
Intuitively, a persistence module $\V$ is $q$-tame any non-diagonal square in the persistence diagram $\PD(\V)$ contains only finitely many of points, see \cite[section~2.8]{chazal2016structure}.  
Any finitely decomposable persistence module is $q$-tame.
  
\begin{thm}[stability of persistence modules]\cite[isometry theorem~4.11]{chazal2016structure}
\label{thm:stability_persistence}
 Let $\mathbb{U}$ and $\mathbb{V}$ be q-tame persistence modules. Then $\ID(\mathbb{U},\mathbb{V}) = \BD(\PD(\mathbb{U}),\PD(\mathbb{V}))$,
 where $\ID$ is the interleaving distance, $\BD$ is the bottleneck distance between persistence modules.
\bs
\end{thm}

\section{Stability of the mergegram for any single-linkage dendrogram}
\label{sec:stability}

In a dendrogram $\De$ from Definition~\ref{dfn:dendrogram}, any merge set $A$ of $\De$ has a life interval $\life(A)=[b,d)$ from its birth scale $b$ to its death scale $d$.
Lemmas~\ref{lem:merge_module_decomposition} and~\ref{lem:merge_modules_interleaved} are proved in appendices. 

\begin{lem}[merge module decomposition]
\label{lem:merge_module_decomposition}
For any dendrogram $\De$ in the sense of Definition~\ref{dfn:dendrogram}, the merge module $M(\De)\cong\bigoplus\limits_{A}\mathbb{I}(\life(A))$ decomposes over all merge sets $A$.
\bs
\end{lem}

Lemma~\ref{lem:merge_module_decomposition} will allow us to use the stability of persistence in Theorem~\ref{thm:stability_persistence} for merge modules and also Lemma~\ref{lem:merge_modules_interleaved}.
Stability of the mergegram $\MG(\De_{SL}(A))$ will be proved under perturbations of $A$ in the Hausdorff distance defined below.

\begin{dfn}[Hausdorff distance HD]
\label{dfn:Hausdorff_distance}
For any subsets $A,B$ of a metric space $(X,d)$, the \emph{Hausdorff distance} $\HD(A,B)$ is $\max\{\sup\limits_{a\in A}\inf\limits_{b\in B} d(a,b), \sup\limits_{b\in B}\inf\limits_{a\in A} d(a,b)\}$.
\bs
\end{dfn}

\begin{lem}[merge modules interleaved]
\label{lem:merge_modules_interleaved}
If any subsets $A,B$ of a metric space $(X,d)$ have $\HD(A,B)=\de$, then the merge modules $M(\De_{SL}(A))$ and $M(\De_{SL}(B))$ are $\de$-interleaved.
\bs
\end{lem}

\begin{thm}[stability of a mergegram]
\label{thm:stability_mergegram}
Any finite subsets $A,B$ of a metric space $(X,d)$ have the mergegrams  $\BD(\MG(\De_{SL}(A)),\MG(\De_{SL}(B))\leq \HD(A,B)$.
Hence any small perturbation of $A$ in the Hausdorff distance yields a similarly small perturbation in the bottleneck distance for its mergegram $\MG(\De_{SL}(A))$ of the single-linkage clustering dendrogram $\De_{SL}(A)$.
\end{thm}
\begin{proof}
The given subsets $A,B$ with $\HD(A,B)=\de$ have $\de$-interleaved merge modules by Lemma~\ref{lem:merge_modules_interleaved}, i.e. $\ID(\MG(\De_{SL}(A)),\MG(\De_{SL}(B))\leq\de$.
Since any merge module $M(\De)$ is finitely decomposable, hence $q$-tame, by Lemma~\ref{lem:merge_module_decomposition}, the corresponding mergegram $\MG(M(\De))$ satisfies Theorem~\ref{thm:stability_persistence}, i.e.
$\BD(\MG(\De_{SL}(A)),\MG(\De_{SL}(B))\leq\de$ as required.
\end{proof}

Theorem~\ref{thm:stability_mergegram} is confirmed by the following experiment on cloud perturbations in Fig.~\ref{fig:BD-vs-noise_bound}.

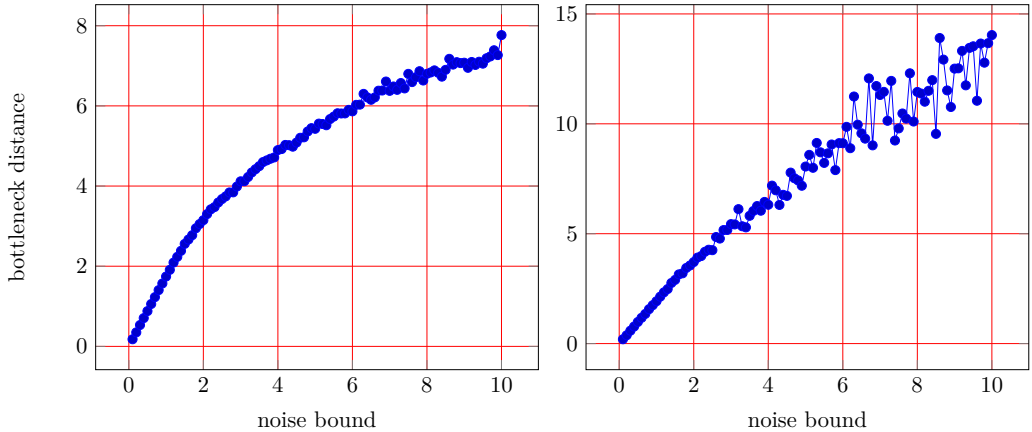
\begin{figure}[h]
\begin{tikzpicture}[scale=0.85]
\begin{axis}[xlabel = noise bound, ylabel = bottleneck distance,grid]
\addplot table [x=a, y=b, col sep=comma] {TableAvg.csv};
\end{axis}
\end{tikzpicture}
\begin{tikzpicture}[scale=0.85]
\begin{axis}[xlabel = noise bound,grid]
\addplot table [x=a, y=b, col sep=comma] {TableMax.csv};
\end{axis}
\end{tikzpicture}
\caption{The bottleneck distances (average on the left, maximum on the right) between mergegrams of sampled point clouds and their perturbations.
Both graphs are below the line $y=2x$. }
\label{fig:BD-vs-noise_bound}
\end{figure}

\begin{enumerate}
\item We uniformly generate $N=100$ black points in the cube $[0,100]^3\subset\R^3$.
\item Then we generate a random number of red points such that the $\epsilon$ ball of every black point randomly has 1, 2 or 3 red points for a noise bound $\epsilon\in[0.1,10]$ taken with a step size 0.1.
\item Compute the bottleneck distance between the mergegrams of black and red points.
  \item Repeat the experiment $K=100$ times, plot the average and maximum in Fig.~\ref{fig:BD-vs-noise_bound}.
\end{enumerate}

\section{Experiments on a classification of point sets and conclusions}
\label{sec:experiments}

Algorithm~\ref{alg:mergegram} computes the mergegram of the SL dendrogram for any finite set $A\subset\R^m$. 

\begin{alg}
\begin{algorithmic}
    \STATE
   \STATE \textbf{Input} : a finite point cloud $A\subset\mathbb{R}^m$
   \STATE Compute $\MST(A)$ and sort all edges of $\MST(A)$ in increasing order of length
   \STATE Initialize Union-Find structure $U$ over $A$. Set all points of $A$ to be their components.
   \STATE Initialize the function $\text{prev: Components}[U] \rightarrow \mathbb{R}$ by setting $\text{prev}(t) = 0$ for all $t$
   \STATE Initialize the vector Output that will consists of pairs in $\mathbb{R} \times \mathbb{R}$
   \FOR{Edge $e = (a,b)$  in the set of edges (increasing order)}
   \STATE Find components $c_1$ and $c_2$ of $a$ and $b$ respectively in Union-Find $U$
   \STATE Add pairs (prev$[c_1]$,length($e$)), (prev$[c_2]$,length($e$))  $\in \mathbb{R}^2$ to Output
   \STATE Merge components $c_1$ and $c_2$ in Union-Find $U$ and denote the component by $t$
   \STATE Set prev$[t]$ = length($e$)
   \ENDFOR
   \STATE \textbf{return} Output
\end{algorithmic}
   \label{alg:mergegram}
\end{alg}

Let $\alpha(n)$ be the inverse Ackermann function.
Other constants below are defined in \cite{march2010fast}.

\begin{thm}[a fast mergegram computation]
\label{thm:complexity}
For any cloud $A\subset\mathbb{R}^m$ of $n$ points, the mergegram $\MG(\De_{SL}(A))$ can be computed in time $O(\max\{c^6,c_p^2c^2_l \}c^{10}n\log n\,\alpha(n))$.
\end{thm}
\begin{proof}
A Minimum Spanning Tree $\MST(A)$ needs $O(\max\{c^6,c_p^2c^2_l \}c^{10}n\log n\,\alpha(n))$ time by  \cite[Theorem~5.1]{march2010fast}. 
The rest of Algorithm~\ref{alg:mergegram} is dominated by $O(n\alpha(n))$ Union-Find operations. 
Hence the full algorithm has the same computational complexity as the MST.
\end{proof}

The experiments summarized in Fig.~\ref{fig:100-point_clouds} show that the mergegram curve in blue outperforms other isometry invariants on the isometry classification by the state-of-the-art PersLay.
We generated 10 classes of 100-point clouds within the unit ball $\R^m$ for $m=2,3,4,5$.
For each class, we made 100 copies of each cloud and perturbed every point by a uniform random shift in a cube of the size $2\times\epsilon$, where $\epsilon$ is called a \emph{noise bound}. 
For each of 100 perturbed clouds, we added 25 points such that every new point is $\epsilon$-close to an original point.
Within each of 10 classes all 100 clouds were randomly rotated within the unit ball around the origin, see Fig.~\ref{fig:clouds}. 
For each of the resulting 1000 clouds, we computed the mergegram, 0D persistence diagram and the diagram of pairs of distances to two nearest neighbors for every point. 

\newcommand{\hh}{44mm}
\begin{figure}[h!]
\includegraphics[height=\hh]{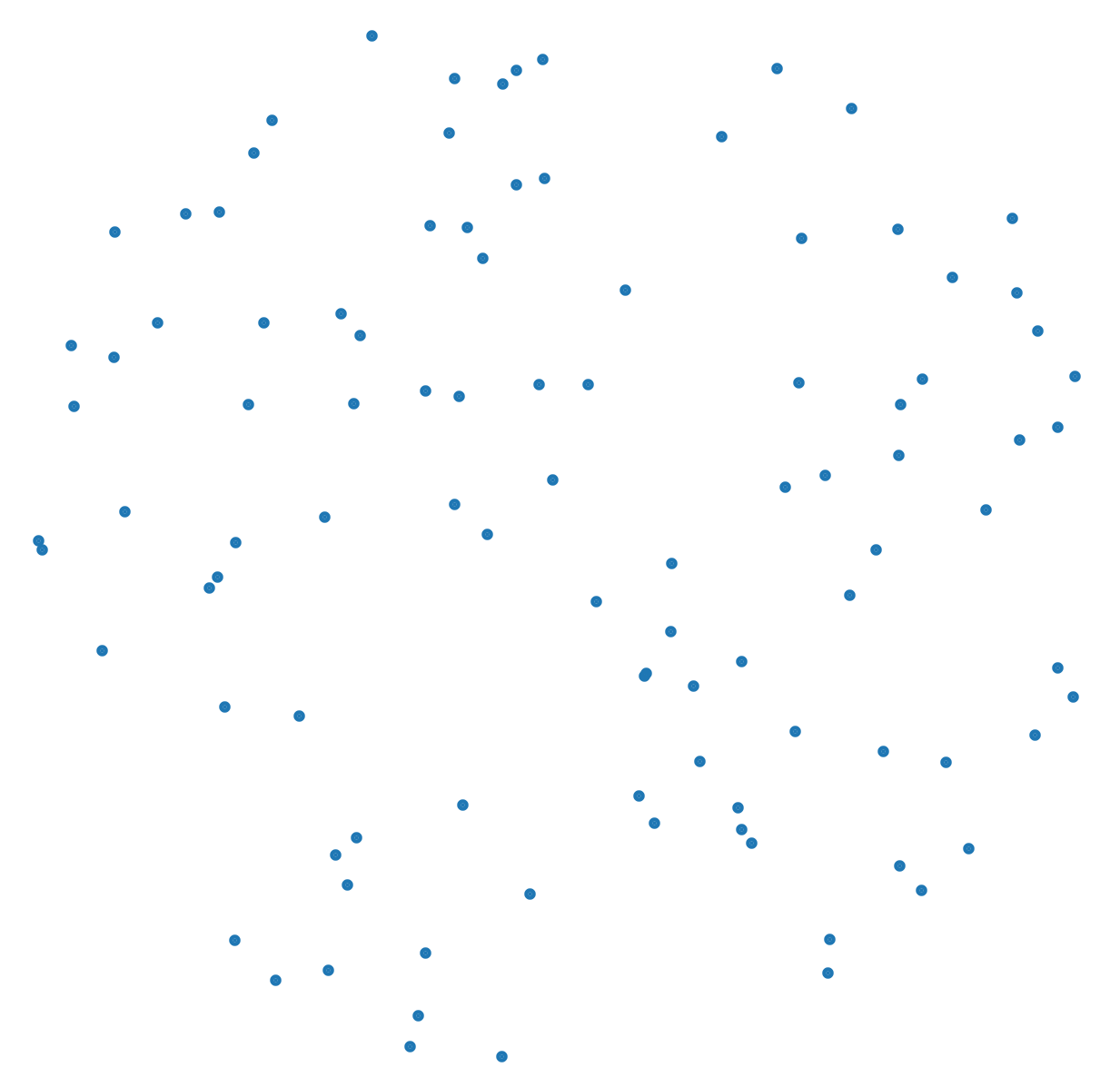}
\includegraphics[height=\hh]{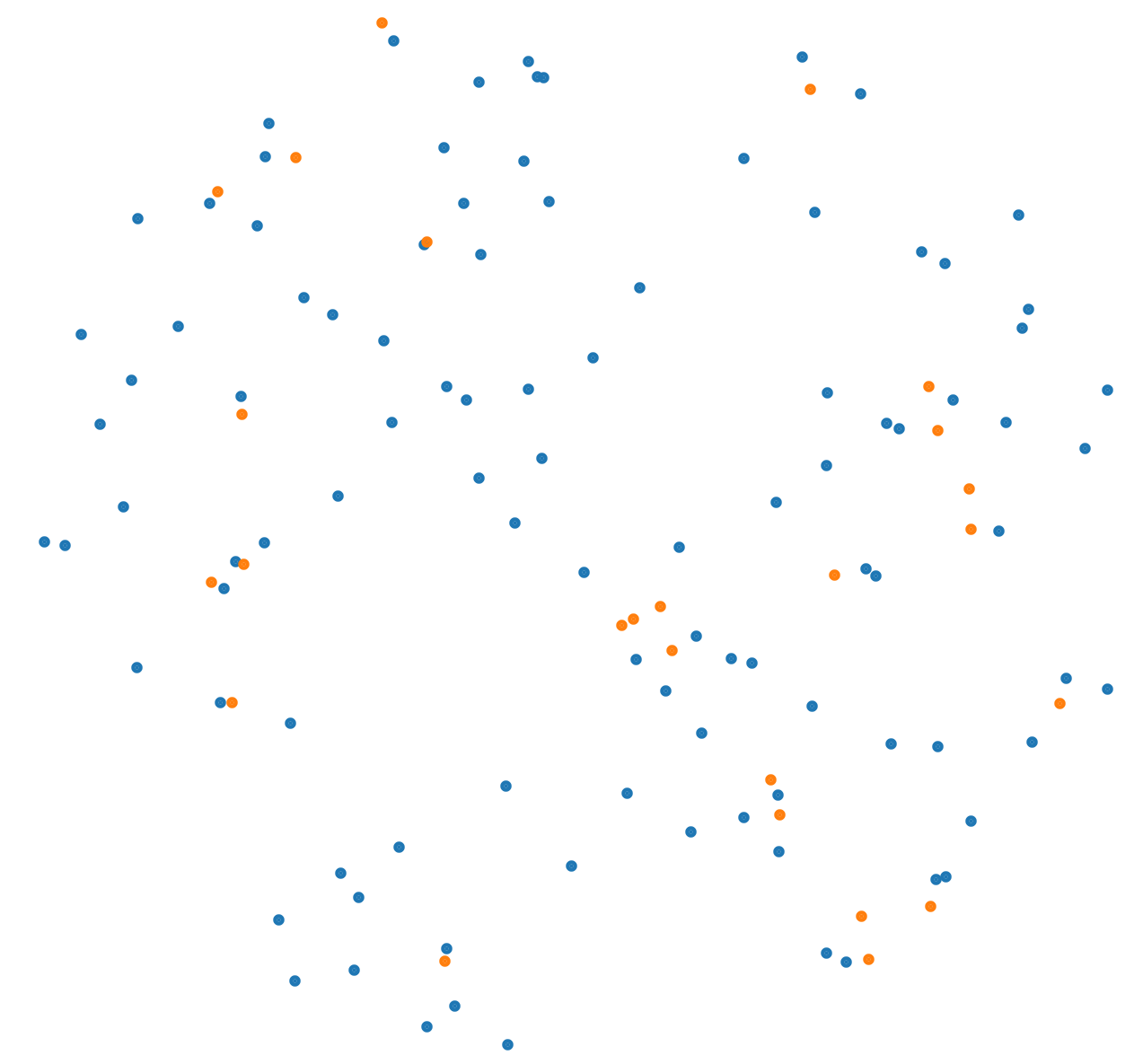}
\includegraphics[height=\hh]{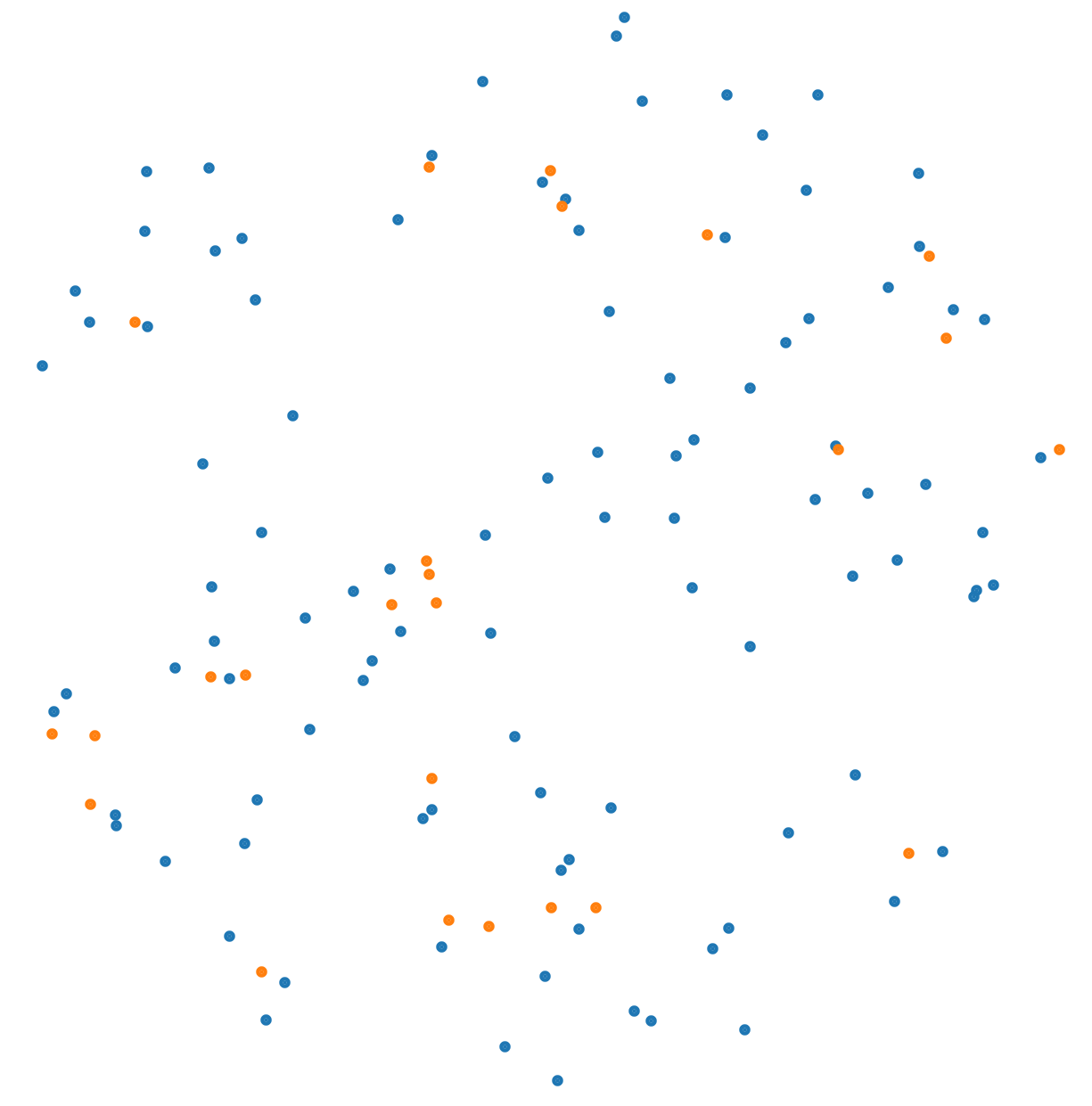}
\caption{\textbf{Left}: an initial random cloud with 100 blue points.
\textbf{Middle}: all blue points are perturbed, 25 extra orange points are added.
\textbf{Right}: a cloud is rotated through a random angle. 
Can we recognize that the initial and final clouds are in the same isometry class modulo small noise?}
\label{fig:clouds}
\end{figure}

The machine learning part has used the obtained diagrams as the input-data for the Perslay \cite{carriere2019perslay}. 
Each dataset was split into learning and test subsets in ratio 4:1. 
The learning loops ran by iterating over mini-batches consisting of 128 elements and going through the full dataset for a given number of epochs. 
The success rate was measured on the test subset.
\medskip

The original Perslay module was rewritten in Tensorflow v2 and RTX 2080 graphics card was used to run the experiments.  
The technical concepts of PersLay are explained in \cite{carriere2019perslay}:

\begin{itemize}
    \item Adam(Epochs = 300, Learning rate = 0.01)
    \item Coefficents = Linear coefficents
    \item Functional layer = [PeL(dim=50), PeL(dim=50, operationalLayer=PermutationMaxLayer)]. 
    \item Operation layer = TopK(50)
\end{itemize}

The PersLay training has used the following invariants compared in Fig.~\ref{fig:100-point_clouds}:
\begin{itemize}
\item cloud : the initial cloud $A$ of points corresponds to the baseline curve in black;
\item PD0: the 0D persistence diagram $\PD$ for distance-based filtrations of sublevel sets in red;
\item NN(2) brown curve: for each point $a\in A$ includes distances to two nearest neighbors;
\item the mergegram $\MG(\De_{SL}(A))$ of the SL dendrogram has the blue curve above others.
\end{itemize}

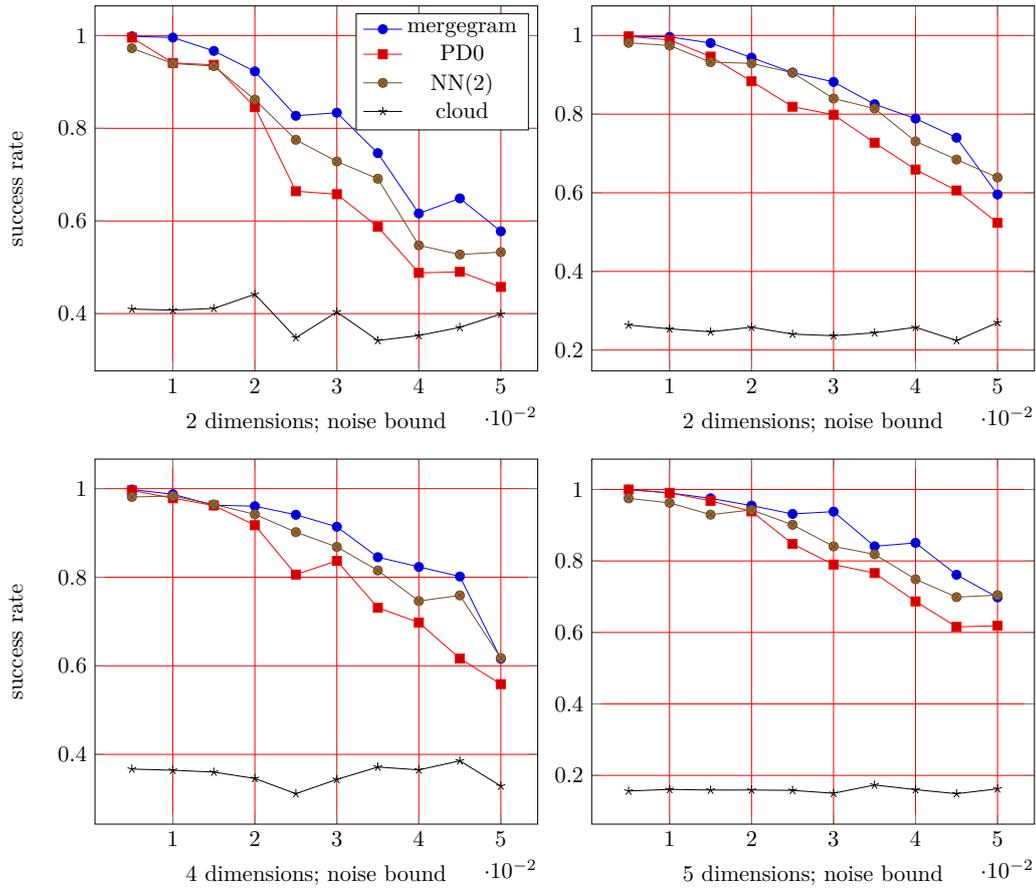
\begin{figure}[t]
\begin{tikzpicture}[scale=0.85]
\begin{axis}[xlabel = 2 dimensions; noise bound, ylabel = success rate,grid]
\addplot table [x=e, y=m, col sep=comma] {100Points25NoiseResults/dim2.csv};
\addlegendentry{mergegram}
\addplot table [x=e, y=h, col sep=comma] {100Points25NoiseResults/dim2.csv};
\addlegendentry{PD0}
\addplot table [x=e, y=n, col sep=comma] {100Points25NoiseResults/dim2.csv};
\addlegendentry{NN(2)}
\addplot table [x=e, y=c, col sep=comma] {100Points25NoiseResults/dim2.csv};
\addlegendentry{cloud}
\end{axis}
\end{tikzpicture}
\begin{tikzpicture}[scale=0.85]
\begin{axis}[xlabel = 2 dimensions; noise bound, grid]
\addplot table [x=e, y=m, col sep=comma] {100Points25NoiseResults/dim3.csv};
\addplot table [x=e, y=h, col sep=comma] {100Points25NoiseResults/dim3.csv};
\addplot table [x=e, y=n, col sep=comma] {100Points25NoiseResults/dim3.csv};
\addplot table [x=e, y=c, col sep=comma] {100Points25NoiseResults/dim3.csv};
\end{axis}
\end{tikzpicture}
\medskip

\begin{tikzpicture}[scale=0.85]
\begin{axis}[xlabel = 4 dimensions; noise bound, ylabel = success rate,grid]
\addplot table [x=e, y=m, col sep=comma] {100Points25NoiseResults/dim4.csv};
\addplot table [x=e, y=h, col sep=comma] {100Points25NoiseResults/dim4.csv};
\addplot table [x=e, y=n, col sep=comma] {100Points25NoiseResults/dim4.csv};
\addplot table [x=e, y=c, col sep=comma] {100Points25NoiseResults/dim4.csv};
\end{axis}
\end{tikzpicture}
\begin{tikzpicture}[scale=0.85]
\begin{axis}[xlabel = 5 dimensions; noise bound,grid]
\addplot table [x=e, y=m, col sep=comma] {100Points25NoiseResults/dim5.csv};
\addplot table [x=e, y=h, col sep=comma] {100Points25NoiseResults/dim5.csv};
\addplot table [x=e, y=n, col sep=comma] {100Points25NoiseResults/dim5.csv};
\addplot table [x=e, y=c, col sep=comma] {100Points25NoiseResults/dim5Cor.csv};
\end{axis}
\end{tikzpicture}
\caption{Success rates of PersLay in identifying isometry classes of 100-point clouds uniformly sampled in a unit ball, averaged over 5 different clouds and 5 cross-validations with 20/80 splits. 
}
\label{fig:100-point_clouds}
\end{figure}

Fig.~\ref{fig:100-point_clouds} shows that the new mergegram has outperformed all other invariants on the isometry classification problem.
The 0D persistence turned out to be weaker than the pairs of distances to two neighbors.
The topological persistence has found applications in data skeletonization with theoretical guarantees \cite{kurlin2015homologically,kalisnik2019higher}. 
We are planning to extend the experiments in section~\ref{sec:experiments} for classifying rigid shapes by comining the new mergegram with the 1D persistence, which has the fast $O(n\log n)$ time for any 2D cloud of $n$ points \cite{kurlin2014fast, kurlin2014auto}.
\medskip

In conclusion, the paper has extended the 0D persistence to a stronger isometry invariant, which has kept the celebrated stability under noise important for applications to noisy data. 
The initial C++ code for the mergregram is at https://github.com/YuryUoL/Mergegram and will be updated.
We  thank all the reviewers for their valuable time and helpful suggestions.

\bibliographystyle{plainurl}
\bibliography{mergegram}

\end{document}